\documentclass[conference]{IEEEtran}
\IEEEoverridecommandlockouts
\usepackage{cite}
\usepackage{amsmath,amssymb,amsfonts}
\usepackage{algorithmic}
\usepackage{graphicx}
\usepackage{textcomp}
\usepackage{xcolor}
\usepackage{amsthm}
\usepackage{bm}
\usepackage{tikz}
\usetikzlibrary{positioning}
\usepackage{algorithmic}
\usepackage{algorithm}
\usepackage{caption}
\usepackage{graphicx}
\usepackage{float} 
\usepackage{subcaption}
\usepackage[final,hyperfootnotes=false,hidelinks]{hyperref}

\newtheorem{thm}{Theorem}

\newtheorem{prop}[thm]{Proposition}

\begin{document}

\title{Data-Driven Neural Estimation of Indirect Rate-Distortion Function\\
\thanks{This work has been supported by the National Natural Science Foundation of China through Grant 62231022.}
}

\author{\IEEEauthorblockN{Zichao Yu, Qiang Sun, and Wenyi Zhang}
      \IEEEauthorblockA{\textit{Department of Electronic Engineering and Information Science} \\
      \textit{University of Science and Technology of China}, \textit{Hefei, China} \\
    \{zichaoyu, qiangsun\}@mail.ustc.edu.cn, wenyizha@ustc.edu.cn}
}

\maketitle

\begin{abstract}
The rate-distortion function (RDF) has long been an information-theoretic benchmark for data compression. As its natural extension, the indirect rate-distortion function (iRDF) corresponds to the scenario where the encoder can only access an observation correlated with the source, rather than the source itself. Such scenario is also relevant for modern applications like remote sensing and goal-oriented communication. The iRDF can be reduced into a standard RDF with the distortion measure replaced by its conditional expectation conditioned upon the observation. This reduction, however, leads to a non-trivial challenge when one needs to estimate the iRDF given datasets only, because without statistical knowledge of the joint probability distribution between the source and its observation, the conditional expectation cannot be evaluated. To tackle this challenge, starting from the well known fact that conditional expectation is the minimum mean-squared error estimator and exploiting a Markovian relationship, we identify a functional equivalence between the reduced distortion measure in the iRDF and the solution of a quadratic loss minimization problem, which can be efficiently approximated by neural network approach. We proceed to reformulate the iRDF as a variational problem corresponding to the Lagrangian representation of the iRDF curve, and propose a neural network based approximate solution, integrating the aforementioned distortion measure estimator. Asymptotic analysis guarantees consistency of the solution, and numerical experimental results demonstrate the accuracy and effectiveness of the algorithm.
\end{abstract}

\section{Introduction}
A natural extension of the lossy source coding problem, named the \textit{indirect} rate distortion problem, was first introduced by Dobrushin and Tsybakov \cite{1057738}, in which the encoder does not have direct access to the source $S$, but instead only has access to an observation $X$, obtained through a probabilistic transformation $P_{X|S}$ from $S$. The objective is to compress $X$ at the smallest possible rate while ensuring that the expected distortion between the source $S$ and its reproduction $Y$ remains below a specified constraint, for a given distortion measure $d: \mathcal{S} \times \mathcal{Y} \rightarrow [0, +\infty]$.

Historically, such indirect rate distortion problem has been used to model the scenario where certain noise is added to the source before encoding. Recently, it has further found extensive applications across various domains, including goal-oriented communication \cite{liu2021rate}, sub-Nyquist sampling \cite{kipnis2015distortion}, privacy protection \cite{isik2022lossy}, estimation from compressed data \cite{kipnis2021rate}, and so on. It facilitates the investigation on the impact of the probabilistic transformation $P_{X|S}$, which can be rather general, on the optimal representation of task-relevant information and its extraction, e.g., the well-known \textit{information bottleneck} paradigm \cite{tishby2000information,goldfeld2020information}. It has been established through a series of information-theoretic studies (see, e.g., \cite{berger2003rate,1054469,1057738,witsenhausen1980indirect}) that, the indirect rate distortion problem can be reduced into an equivalent standard rate distortion problem with a transformed distortion measure $\Bar{d}: \mathcal{X} \times \mathcal{Y} \rightarrow [0, +\infty]$ given by
\begin{align}\label{eqn:d-bar}
    \Bar{d}(x,y) = E[d(S,y)|X=x].
\end{align}
Intuitively, since $S$ (and thus $d(S, y)$) is not observable, when performing coding, one can only work with some form of estimate of $d(S, y)$, and the conditional expectation $\Bar{d}(x,y)$ turns out to be the choice. Further information-theoretic aspects of the indirect rate distortion problem, such as non-asymptotic and multi-terminal settings, have been studied in \cite{yamamoto1980source,490552,oohama2012distributed,guler2015remote,7464359,eswaran2019remote} and references therein.
The information-theoretic fundamental trade-off of the indirect rate distortion problem is characterized by the \textit{indirect} rate-distortion function (iRDF). As discussed in the preceding paragraph, the iRDF can be obtained as the RDF of an equivalent standard rate distortion problem. Analytical solutions of iRDF have been available for a few specific models only; see, e.g., \cite{kipnis2015indirect,1057738,1054469,liu2021rate,9844779}. Although for general models one can employ standard numerical algorithms such as the Blahut-Arimoto algorithm and its variants \cite{1054855,chen2023constrained,hayashi2023bregman}, in order to cope with modern applications where the statistical knowledge of the source and its observation may not be readily available, it is also of great interest to develop data-driven algorithms.

To date, there has been some work on data-driven estimation of standard RDF \cite{10206867,yang2022towards,lei2022neural}, typically based on neural network approximation of some variational characterization of mutual information. Although the iRDF can be reduced into an equivalent RDF, when considering data-driven estimation, however, a non-trivial challenge arises. This is because evaluating the transformed distortion measure $\Bar{d}$ in \eqref{eqn:d-bar} requires the statistical knowledge of \textit{backward} conditional probability distribution $P_{S|X}$, which is unfortunately unavailable under a data-driven setting. In principle, with sufficient amount of data samples, one may construct an empirical estimate of $P_{S|X}$ and use it to obtain an empirical estimate of $\Bar{d}$. This, however, is usually infeasible in practice, as conditional density estimation has been a notoriously difficult problem for high-dimensional settings. Consider, for example, a typical goal-oriented communication setting where the source $S$ is a label (e.g., a dog or cat), $X$ is an image corresponding to $S$ (e.g., a photo of a dog or cat), and the reproduction $Y$ is a classification result based on a compressed version of $X$. For a given label $S$, a dataset may contain a large number of samples of $X$ so as to provide a reasonable empirical estimate of the \textit{forward} conditional probability distribution $P_{X|S}$, but for a given image $X$ there is no clear way to accurately sampling the \textit{backward} conditional probability distribution $P_{S|X}$. In fact, sampling from the posterior (i.e., backward conditional probability distribution) and estimating it holds significant importance in Bayesian inference. Existing methods, such as Markov Chain Monte Carlo (MCMC) \cite{brooks1998markov}, usually face challenges related to issues such as complexity of integration in high-dimensional spaces and sampling precision.

In this work, we propose a data-driven approach to estimating iRDF. In order to tackle the challenge discussed in the preceding paragraph, we first recognize that the transformed distortion measure $\Bar{d}$, as a conditional expectation, corresponds to the solution of a corresponding minimum mean-squared error (MMSE) estimation problem, i.e., for each $y$, estimating $d(S, y)$ upon observing $X = x$. A naive regression estimate based on quadratic loss minimization, however, turns out to be inefficient and even infeasible, because it requires a separate estimate for each $y \in \mathcal{Y}$. Instead, via identifying an alternative quadratic loss minimization problem by exploiting the Markov chain $S \leftrightarrow X \leftrightarrow Y$, we provide a technique for efficiently sampling $\Bar{d}(X, Y)$.

We subsequently reformulate the iRDF as a variational problem corresponding to the Lagrangian representation of the iRDF curve, and solve it leveraging the mapping approach \cite{rose1994mapping}, which transforms optimization in probability space into an equivalent problem of finding optimal mapping. In implementation, we design an algorithm adopting a dual-layer nested neural network, in which the outer layer searches for the optimal mapping and the inner layer estimates $\Bar{d}$. Both layers are progressively trained throughout epochs. We conduct an asymptotic analysis of the algorithm to guarantee its consistency, and also conduct extensive numerical experiments on synthetic and MNIST datasets to demonstrate the accuracy and effectiveness of the algorithm.
\section{Problem Formulation}\label{sec:problem}
Consider a lossy source coding setup illustrated in Figure \ref{fig:model},
where the source \(S\), observation \(X\), reproduction \(Y\), and distortion measure $d$ have been introduced in the previous section.
According to \cite{1057738,berger2003rate,1054469,witsenhausen1980indirect}, the iRDF is given as follows:
\begin{align}
    R(D) &= \min_{\substack{P_{Y|X}: \\ E[\Bar{d}(X,Y)] \leq D }} I(X;Y), \label{ird}
\end{align}
where the conditional probability distribution $P_{Y|X}$ satisfies the expected distortion constraint with respect to the reduced distortion measure $\Bar{d}(x,y) = E[d(S,y)|X=x]$. In this paper, we focus on estimating $R(D)$, in the context where the statistical knowledge of model, i.e., the joint probability distribution $P_{S, X}$ is unknown, except for a dataset of $n$ independent and identically distributed (i.i.d.) samples drawn from $P_{S, X}$, realized as $(x_i, s_i)_{i=1}^{n}$.

\begin{figure}[ht]
    \centering
    \begin{tikzpicture}[>=latex, transform shape]
        \node (s0) at (-1,0) [circle, draw=black,fill=gray] {Source};
        \node (s1) at (1,0) [rectangle] {$P_{X|S}$};
        \node (s2) at (3,0) [rectangle, draw=black,fill=gray!10] {Encoder};
        \node (s3) at (5,0) [rectangle, draw=black,fill=gray!10] {Decoder};
        \node (s4) at (7,0) {};
        \draw[->] (s0) to  node[above]{$S$} (s1) ;
        \draw[->] (s1) to  node[above]{$X$} (s2) ;
        \draw[->] (s2) to  (s3) ;
        \draw[->] (s3) to node[above]{$Y$} (s4) ;
    \end{tikzpicture}
    \caption{Indirect rate-distortion model}
    \label{fig:model}
\end{figure}
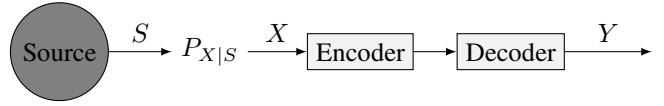
To make the problem tractable, due to the convexity and non-increasing property of $R(D)$ \cite{berger2003rate}, we transform $R(D)$ into an equivalent unconstrained Lagrangian optimization problem \cite{csiszar2011information}, parameterized by the Lagrangian multiplier $\lambda$, as
\begin{align}
   F(\lambda) & = \inf_{P_{Y|X}} I(X;Y) + \lambda E[\Bar{d}(X,Y)]  \label{l}
\end{align}
According to \cite{csiszar2011information}, we have
\begin{align}
    R(D)=\sup_{\lambda \geq 0}[F(\lambda)-\lambda D]. \label{3}
\end{align}
So geometrically $F(\lambda)$ is the $R$-axis intercept of a tangent line to the iRDF $R(D)$ with slope $-\lambda$, and conversely,
$R(D)$ is determined by the convex envelope of all such tangent lines.
The tangent point $(D_{\lambda},R_{\lambda})$ associated with each fixed $\lambda \geq 0$ is achieved by the optimal solution $P_{Y|X}^*$ in \eqref{l}, i.e., $R_{\lambda}=I(X;Y^*)$ and $D_{\lambda}=E[\Bar{d}(X,Y^*)]$, where $Y^*$ denotes the reproduction induced by $P_{Y|X}^*$.
Utilizing the golden formula for relative entropy (see, e.g., \cite[Thm. 4.1]{polyanskiy_wu_2024}), we can rewrite the mutual information $I(X; Y)$ as $\inf_{Q_Y} D(P_{Y|X}\|Q_Y |P_X)$, and $F(\lambda)$ in \eqref{l} becomes
\begin{align}\label{eqn:F_lambda_golden}
    F(\lambda) & = \inf_{Q_{Y},P_{Y|X}} D(P_{Y|X}\| Q_{Y} |P_X) + \lambda E[\Bar{d}(X,Y)].
\end{align}
Under given $Q_Y$, the minimizer $P^*_{Y|X}$ of \eqref{eqn:F_lambda_golden} can be found, by equating the variation of $F(\lambda)$ to zero, as
\begin{align}\label{eqn:P_minimizer}
    P^*_{Y|X}(y|x) = \frac{1}{Z(x)}Q_{Y}(y)\mathrm{e}^{-\lambda \Bar{d}(x,y)},
\end{align}
where $Z(x) = E_{Q_Y} \left[e^{-\lambda \Bar{d}(x, Y)}\right]$ is a normalization parameter. Plugging \eqref{eqn:P_minimizer} into \eqref{eqn:F_lambda_golden}, we can deduce $F(\lambda)$ into
\begin{align}
    F(\lambda) & = \inf_{Q_{Y}} -E_{P_X}\left[\log E_{Q_{{Y}}} \left[\mathrm{e}^{-\lambda \Bar{d}(X,Y)}\right]\right]. \label{F:lambda}
\end{align}

Once obtaining the minimizer $Q^*_Y$ of \eqref{F:lambda}, we get
\begin{align}
    D_{\lambda} = E_{P_X \otimes Q^*_{Y}} \left[ \Bar{d}(X,Y) \frac{\exp\left(-\lambda \Bar{d}(X,Y)\right)}{E_{Y' \sim Q^*_{Y}} \left[\exp\left(-\lambda \Bar{d}(X,Y')\right)\right]} \right]  \label{D_s}
\end{align}
and
\begin{align}
R_{\lambda} = -E_{P_X}\left[\log E_{Q^*_{{Y}}} \left[\mathrm{e}^{-\lambda \Bar{d}(X,Y)}\right]\right] - \lambda D_{\lambda}.
\end{align}
By sweeping over $\lambda \geq 0$, we thus obtain the entire iRDF curve.
\section{Neural Estimator and Consistency Analysis}

In this section, we develop a data-driven neural estimator for iRDF, termed Neural Estimator of Indirect Rate Distortion (NEIRD) algorithm, and provide its consistency analysis.

\subsection{Characterization of Reduced Distortion Measure}

As discussion in the introduction, a key challenge for data-driven estimation of iRDF is the lack of statistical knowledge which hinders the evaluation of the reduced distortion measure $\Bar{d}(x, y)$, as the conditional expectation of $d(S, y)$ upon observing $X = x$. Instead of the highly inefficient naive approach of forming an empirical estimate of $\Bar{d}(x, y)$ for each possible $y$, the following proposition reveals that one can generate $Y$ as an auxiliary random variable, and solve a quadratic loss minimization problem to obtain an estimate of $\Bar{d}(x, y)$ for all possible $(x, y) \in \mathcal{X} \times \mathcal{Y}$ altogether at once.

\begin{prop}
    Considering general random variables $S$, $X$ and $Y$ that form a Markov chain $S \leftrightarrow X \leftrightarrow Y$, we have
    \begin{align}
        \mathop{\arg\min}_{g} E \left[\left(d(S,Y) - g(X,Y)\right)^2 \right] = \Bar{d}, \label{conditional} 
    \end{align}
    where $g$ is among all measurable functions over $\mathcal{X}\times \mathcal{Y}$.
\end{prop}
\begin{proof}
    We begin with the well known fact that conditional expectation solves the MMSE estimation problem and note that the left side of \eqref{conditional} seeks for an estimator of $d(S, Y)$ upon observing $(X, Y)$, i.e.,
    \begin{align}
        g(x, y) = E[d(S, y)|x, y].
    \end{align}

    Utilizing the Markov chain $S \leftrightarrow X \leftrightarrow Y$, we get
    \begin{align}
        g(x, y) = E[d(S, y)|x] = \Bar{d}(x, y).
    \end{align}

    Since this holds for every $(x, y) \in \mathcal{X} \times \mathcal{Y}$, we establish \eqref{conditional}.
\end{proof}

Despite of its simplicity, the proposition has a few consequences for developing the NEIRD algorithm. First, \eqref{conditional} indicates that the minimizer of the quadratic loss on its left side is exactly the desired reduced distortion measure, and this holds over the entirety of $\mathcal{X} \times \mathcal{Y}$, not merely any particular $(x, y) \in \mathcal{X} \times \mathcal{Y}$. This implies that we can estimate $\Bar{d}$ as a function, instead of estimating it for each individual $(x, y)$ one by one. Second, there is no restriction on the probability distribution of $Y$, except for the Markov chain $S \leftrightarrow X \leftrightarrow Y$. So we can view $Y$ as an auxiliary random variable and generate it from any probability distribution at convenience. This greatly reduces the cost of generating datasets. Motivated by these, we can introduce a set of neural network parameters $\hat{\Theta}$ to implement $g$ as a neural network denoted by $g_{\hat{\theta}}$, and train $g_{\hat{\theta}}$ using standard gradient descent techniques so as to empirically minimize the quadratic loss in \eqref{conditional}. Once the training is completed, such neural network $g_{\hat{\theta}}$ can be used to generate samples of $\Bar{d}(X, Y)$.

\subsection{NEIRD Algorithm}

For the problem in Section \ref{sec:problem}, instead of directly optimizing the distribution $Q_Y$, we adhere to the idea of optimizing over a family of measurable mappings $T_{\theta}: \mathcal{Z} \to \mathcal{Y}$, which is parameterized by a neural network with parameters $\theta \in \Theta$. For results on the complexity issue of the neural network $T_{\theta}$, see, e.g.,\cite{lu2020universal}. This way, the distribution $Q_Y$ is the image measure induced by $T_{\theta}$, namely, $T_{\theta} \circ P_Z^{-1}$ where $P_Z$ is some basis probability distribution (such as a uniform distribution over a unit cube or a multivariate standard Gaussian distribution) and $\circ$ denotes the push-forward operator. The mapping approach has been applied to the evaluation of standard RDF in \cite{rose1994mapping}, and the existence of such measurable mapping is guaranteed by the Borel isomorphism theorem (see, e.g., \cite{royden2010real}), from which we have that the measurable spaces $(\mathcal{Z}, P_Z)$ and $(\mathcal{Y}, Q_{Y})$ are isomorphic to each other because they are both isomorphic to the $[0,1]$ interval with Lebesgue measure.

Denoting the objective in $F(\lambda)$ \eqref{F:lambda} by $\mathcal{L}^{\lambda}(Q_Y)$, it can be further rewritten via the mapping approach as
\begin{align}
    \mathcal{L}^{\lambda}(Q_Y)= -E_{P_X}\left[\log E_{Z}\left[\mathrm{e}^{-\lambda \Bar{d}(X,T_{\theta}(Z))}\right]\right].\label{14}
\end{align}

Evaluating these expectations via high-dimensional integration is infeasible without exact statistical knowledge, so it is foreseeable that empirical expectation will be employed as the approximate solution. Given i.i.d. datasets $(x_i, s_i)_{i = 1}^n$ and $(z_j)_{j = 1}^m$ drawn from $P_{S, X}$ and $P_Z$, respectively, the expectations in \eqref{14} and \eqref{D_s} can be approximated as
\begin{align}
  \widehat{\mathcal{L}}^{\lambda}(\theta, \hat{\theta}) = -\frac{1}{n} \sum_{i=1}^{n}\log \left(\frac{1}{m}\sum_{j=1}^m \mathrm{e}^{-\lambda g_{\hat{\theta}}(x_i,T_{\theta}(z_j)) }\right)
\end{align}
and
\begin{align}
    \widehat{D}({\theta}, \hat{\theta}) = \frac{1}{mn}\sum_{i=1}^m\sum_{j=1}^{n}g_{\hat{\theta}}(x_i,T_{\theta}(z_j))\frac{\mathrm{e}^{-\lambda g_{\hat{\theta}}(x,T_{\theta}(z_j))}}{\sum_{j=1}^n\mathrm{e}^{-\lambda g_{\hat{\theta}}(x,T_{\theta}(z_j))}}.
\end{align}
These lead to our empirical estimate of $R(D)$ as
\begin{align}
    \widehat{R}(D)=\widehat{\mathcal{L}}^{\lambda}(\theta, \hat{\theta}) - \lambda \widehat{D}({\theta}, \hat{\theta}).
\end{align}

Detailed steps of the NEIRD algorithm are summarized in Algorithm \ref{alg}.
\subsection{Consistency Analysis}

Next, we analyze the consistency of the NEIRD algorithm. In a nutshell, NEIRD relies on two neural networks $g_{\hat{\theta}}$ and $T_\theta$, as well as datasets $(x_i, s_i)_{i = 1}^n$ and $(z_j)_{j = 1}^m$. Therefore, the consistency analysis is essentially divided into two issues: one is the \textit{approximation problem} related to the neural network parameter spaces $\widehat{\Theta}$ and $\Theta$, and the other is the \textit{estimation problem} related to the use of empirical measures.

\begin{thm}
Under each fixed $\lambda \geq 0$ and any given $\epsilon > 0$, there exist $\theta \in \Theta$ and $\hat{\theta} \in \hat{\Theta}$ (where $\Theta \in \mathbb{R}^k$ and $\hat{\Theta} \in \mathbb{R}^k$, for some sufficiently large $k$, are compact sets), and $N$ and $M$ (possibly depending upon $\theta$ and $\widehat{\theta}$), such that
\begin{align}
    \forall n\geq N, m\geq M ,  |\widehat{R}(D) - R(D)|\leq \epsilon, a.e.; \label{thm3}
\end{align}
that is, NEIRD is strongly consistent.
\end{thm}
\begin{proof}
Denote the minimizers of $\eqref{conditional}$ and \eqref{F:lambda} by $g^*(x,y)$ and $Q_Y^*$, respectively.
First, we prove that there exist $Q_{\theta} \sim T_{\theta}(Z)$ where $\theta \in \Theta$ and $g_{\widehat{\theta}}$ where $\hat{\theta} \in \Hat{\Theta}$ such that
\begin{align}
    |R(Q_Y^*, g^*)- R(Q_{\theta}, g_{\hat{\theta}})|\leq\frac{\epsilon}{2}. 
\end{align}
Through the triangle inequality, we have
\begin{align}
    &|R(Q_Y^*,g^*)-R(Q_{\theta}, g_{\hat{\theta}})| \nonumber\\
    \leq &{|R(Q_Y^*,g^*)-R(Q_{\theta}, g^*)|} + {|R(Q_{\theta},g^*)-R(Q_{\theta}, g_{\hat{\theta}})|}, \label{1}
\end{align}
and it is sufficient to show that both terms on the right side of \eqref{1} are less than $\frac{\epsilon}{4}$. According to the universal approximation theorem \cite{hornik1989multilayer}, there exists some $\hat{\theta} \in \hat{\Theta}$ such that $\forall (x, y) \in \mathcal{X} \times \mathcal{Y}$,
\begin{align}
    |g_{\hat{\theta}}(x,y) - g^*(x,y)|\leq \frac{\mathrm{e}^{\frac{\epsilon}{4}}}{\lambda}. \label{2}
\end{align}
Since $\mathrm{e}^x$ is 1-Lipschitz continuous on the interval $(-\infty, 0]$, and $\log x$ is continuous, it holds that
\begin{align}
     \log \left|E_{Q_{\theta}}\left[{e}^{-\lambda \Bar{d}(x,y)} - \mathrm{e}^{-\lambda g_{\hat{\theta}}(x,y)}\right]\right|
         \leq \log  \mathrm{e}^{\frac{\epsilon}{4}}=\frac{\epsilon}{4},
\end{align}
and this ensures that $|\mathcal{L}^{\lambda}(Q_{\theta},g^*) - \mathcal{L}^{\lambda}(Q_{\theta},g_{\hat{\theta}})|\leq \frac{\epsilon}{4}$.

Regarding \eqref{D_s}, we rewrite it as follows:
\begin{align}
D_{\lambda}(g^*)= E_{P_X} \left[ \frac{ E_{Q_{\theta}}[g^*(X,Y) \mathrm{e}^{-\lambda g^*(X,Y)}] }{E_{Q_{\theta}} \left[\mathrm{e}^{-\lambda g^*(X,Y)}\right]} \right]. \label{pr1}
\end{align}
Since $xe^x$ is uniformly continuous on the interval $(-\infty,0]$ and $\mathrm{e}^x$ is 1-Lipschitz continuous on the interval $(-\infty, 0]$, there exists $\upsilon>0$, such that 
    \begin{align}
        |D_{\lambda}(g^*)-D_{\lambda}(g_{\hat{\theta}})|< \frac{\epsilon}{4 \lambda},
    \end{align}
if $|g_{\hat{\theta}}(x,y) - g^*(x,y)|\leq \upsilon$.
Thus we can conclude that $|R(Q_{\theta},g^*)-R(Q_{\theta}, g_{\hat{\theta}})|<\frac{\epsilon}{4}$.
Next, by following the same argument as that in \cite{lei2022neural}, we can obtain $|R(Q_Y^*,g^*)-R(Q_{\theta}, g^*)| \leq \frac{\epsilon}{4}$. 
Therefore, we conclude that 
        \begin{align}
            |R(D)-R(Q_{\theta},g_{\hat{\theta}})| \leq \frac{\epsilon}{2}, \label{thm1}
        \end{align}
        
We further proceed as
        \begin{align}
            &\left|\widehat{R}(D)- \inf R(Q_{\theta},g_{\hat{\theta}})\right| \nonumber\\
            &=\left|\inf \{ \widehat{\mathcal{L}}({\theta}, \widehat{\theta})-\lambda \widehat{D}({\theta}, \widehat{\theta}) \} - \inf R(Q_{\theta},g_{\hat{\theta}})\right| \nonumber\\
            &\leq \sup \left|\widehat{R}(Q_{\theta},g_{\hat{\theta}})-R(Q_{\theta},g_{\hat{\theta}})\right|,
         \end{align}
where $\widehat{R}(Q_{\theta},g_{\hat{\theta}})=\widehat{\mathcal{L}}({\theta}, \hat{\theta})-\lambda \widehat{D}({\theta}, \hat{\theta})$.
         Since $\Theta$ and $\hat{\Theta}$ are compact and feed-forward neural network functions are continuous, the family of functions $Q_{\theta}$ and $g_{\hat{\theta}}$ satisfy the uniform law of
large numbers \cite{geer2000empirical}. Thus, $\forall \epsilon>0$, there exist $N$ and $M$ such that $\forall n\geq N, m\geq M $, with probability one, $\left|\widehat{R}(Q_{\theta},g_{\hat{\theta}})-R(Q_{\theta},g_{\hat{\theta}})\right| \leq \frac{\epsilon}{2}$. We hence obtain
\begin{align}
   |\widehat{R}(D)-R(Q_{\theta},g_{\hat{\theta}})|\leq\frac{\epsilon}{2}, a.e. \label{thm2}
\end{align}
Combining \eqref{thm1} and \eqref{thm2}, we arrive at 
\begin{align}
\left|\widehat{R}(D)-R(D)\right| \leq \epsilon, \forall n\geq N, m\geq M .
\end{align}
The proof is thus completed.  
\end{proof}
\begin{algorithm}[ht] \label{alg}
    \caption{Neural Estimator of Indirect Rate Distortion}
    \label{alg}
    \renewcommand{\algorithmicrequire}{\textbf{Input:}}
    \renewcommand{\algorithmicensure}{\textbf{Output:}}
    \begin{algorithmic}[1]
        \REQUIRE Slope $\lambda$, dataset sizes $n$ and $m$, minibatch size $b$, learning rates $\eta$ and $\xi$, number of iteration \textit{epochs} and $t$, initial neural network parameters $\theta$ of $T_{\theta}$ and $\hat{\theta}$ of $g_{\hat{\theta}}$
                \FOR{$\ell = 1$ to $\textit{epochs}$}
                \STATE Sample $\{x_i\}_{i=1}^n$ i.i.d. $\sim P_X$
                \STATE Sample $\{z_j\}_{j=1}^m$ i.i.d. $\sim P_Z$
                \STATE Calculate $\kappa_{ij}(\lambda,\theta) =\mathrm{e}^{-\lambda g_{\hat{\theta}}(x_i,T_{\theta}(z_j))}$
                \STATE $\theta \leftarrow \theta - \eta \nabla_{\theta} (-\frac{1}{n}\sum_{i=1}^{n}\log (\frac{1}{m}\sum_{j=1}^m \kappa_{ij}(\lambda, \theta) )$
                \FOR{$t'= 0$ to $t$}
                \STATE Sample minibatch $\{(x_i,s_i) \}_{i=1}^b$
                \STATE Sample auxiliary minibatch $\{y_i\}_{i = 1}^b$ 
                \STATE $\hat{\theta} \leftarrow \hat{\theta} - \xi \nabla_{\hat{\theta}} (\frac{1}{b}\sum_{i=1}^b\|g_{\hat{\theta}}(x_i,y_i)-d(s_i,y_i) \|^2)$
                \ENDFOR
        \ENDFOR
        \STATE Calculate $\widehat{R}(D)=\widehat{\mathcal{L}}^{\lambda}(\theta, \hat{\theta}) - \lambda \widehat{D}({\theta}, \hat{\theta})$
        \STATE \textbf{Return:} $g_{\hat{\theta}}$, $T_{\theta}$, and $\widehat{R}(D)$
    \end{algorithmic}
\end{algorithm}

\section{Numerical Experiments}
In this section, we validate the effectiveness of the proposed NEIRD algorithm on both synthetic datasets with known theoretical solutions and the MNIST dataset. 
\subsection{Synthetic Datasets}

We consider two cases: jointly Gaussian model from \cite{9844779} and binary classification model from \cite{liu2021rate}. For both synthetic datasets, the networks consists of five fully-connected layers.

\begin{figure*}[h!]
    \centering
    \begin{subfigure}[b]{0.3\textwidth}
        \centering
        \includegraphics[width=\textwidth]{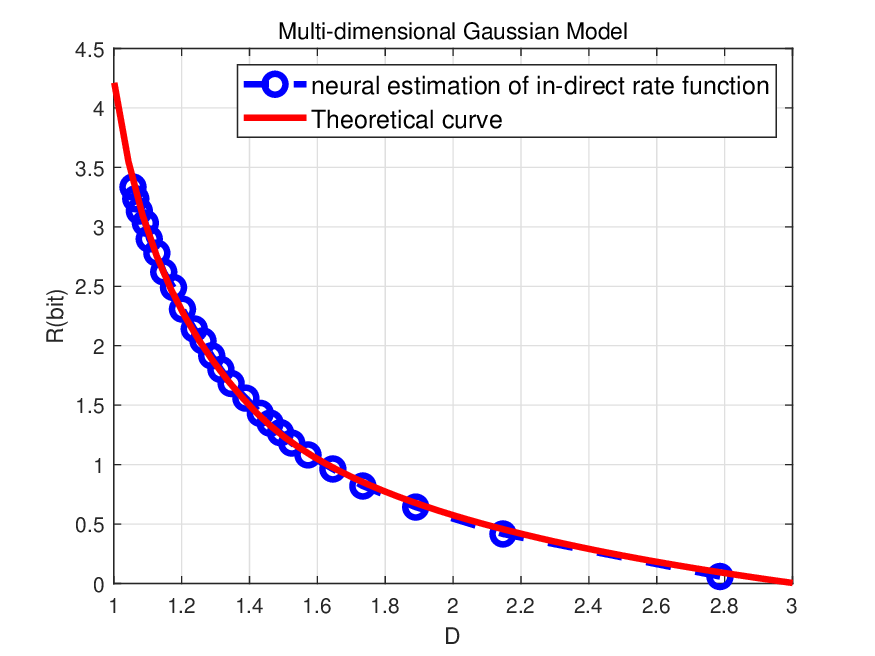}
        \caption{Low dimensional jointly Gaussian}
        \label{fig:fig1}
    \end{subfigure}
    \hfill
    \begin{subfigure}[b]{0.3\textwidth}
        \centering
        \includegraphics[width=\textwidth]{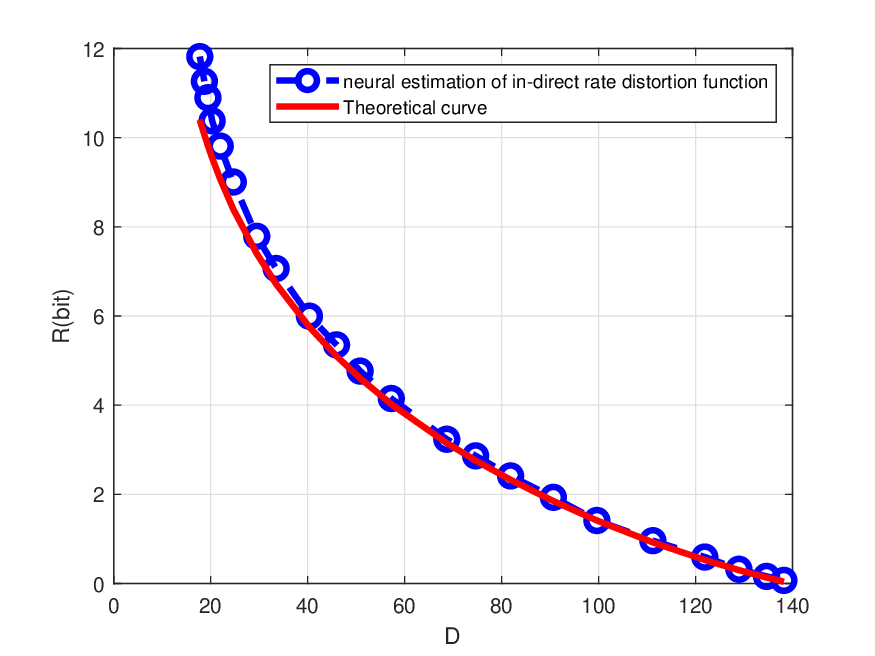}
        \caption{High dimensional jointly Gaussian}
        \label{fig:fig2}
    \end{subfigure}
    \hfill
    \begin{subfigure}[b]{0.3\textwidth}
        \centering
        \includegraphics[width=\textwidth]{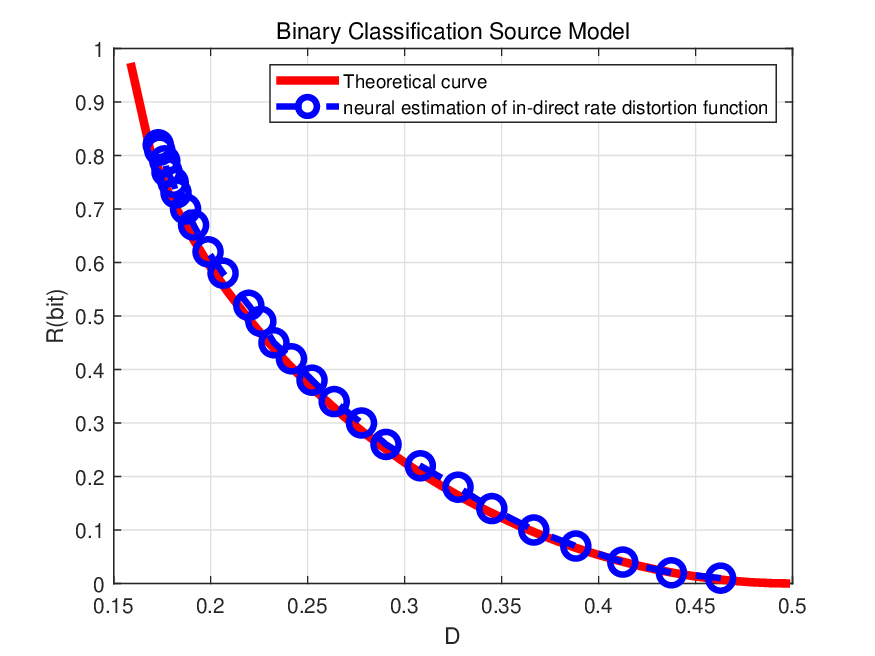}
        \caption{Binary classification}
        \label{fig:fig3}
    \end{subfigure}
    \caption{Performance of NEIRD on synthetic datasets.}
    \label{fig:synthetic source}
\end{figure*}
\begin{figure*}[h!]
    \centering
    \begin{subfigure}[b]{0.3\textwidth}
        \centering
        \includegraphics[width=\textwidth]{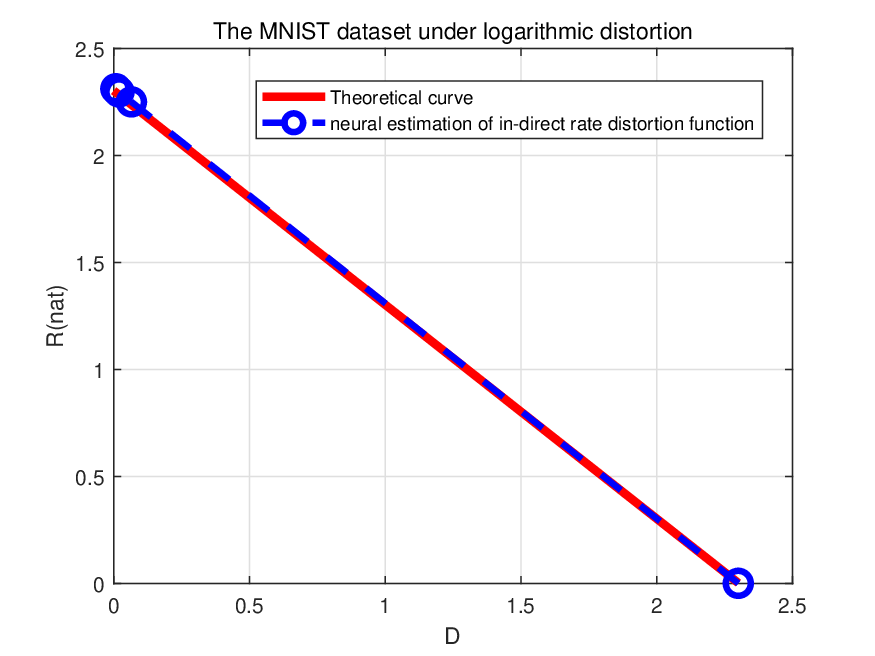}
        \caption{Logarithmic loss distortion}
        \label{fig:fig4}
    \end{subfigure}
    \hfill
    \begin{subfigure}[b]{0.3\textwidth}
        \centering
        \includegraphics[width=\textwidth]{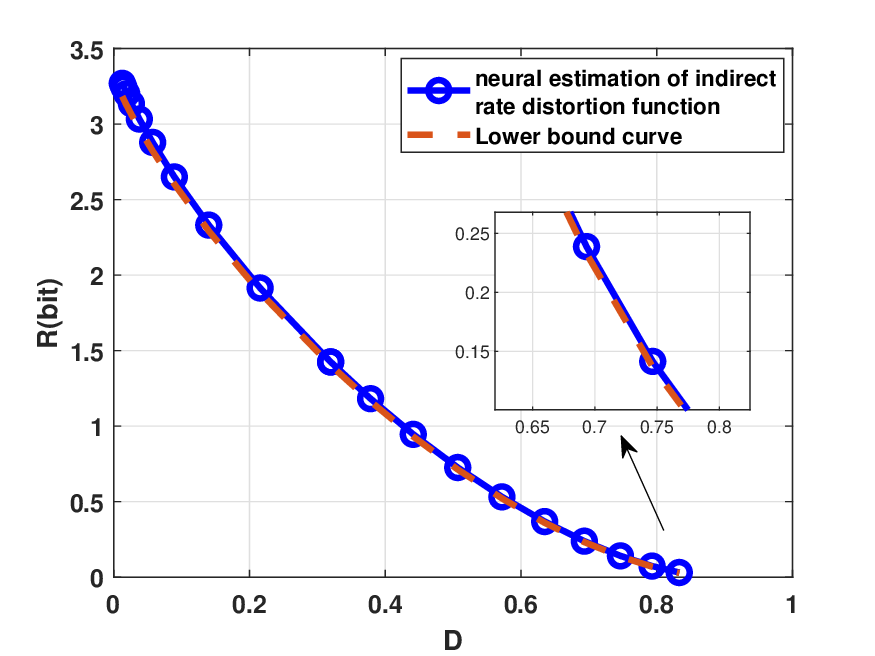}
        \caption{Hamming distortion}
        \label{fig:fig5}
    \end{subfigure}
    \hfill
    \begin{subfigure}[b]{0.3\textwidth}
        \centering
        \includegraphics[width=\textwidth]{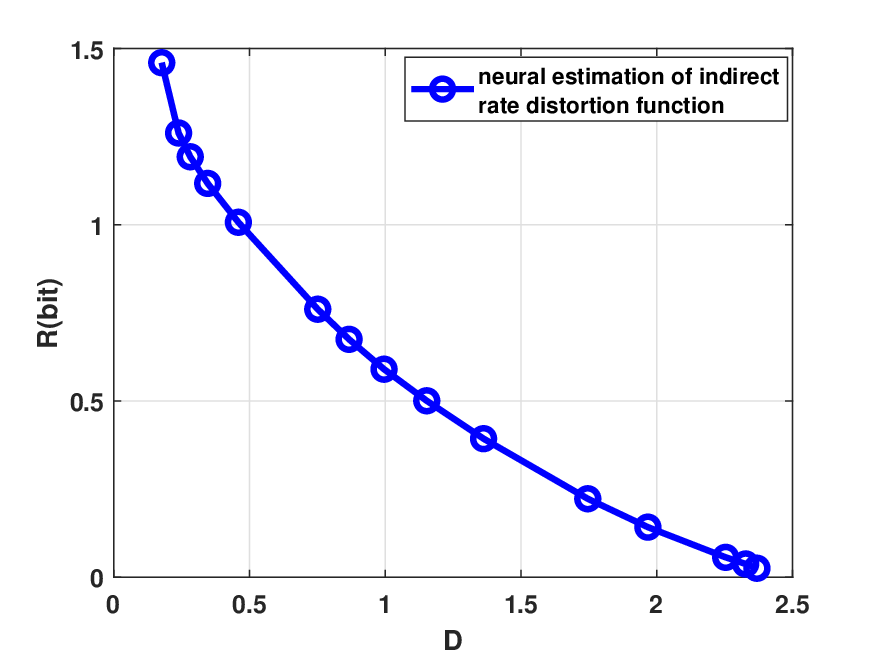}
        \caption{Cross-entropy distortion}
        \label{fig:fig6}
    \end{subfigure}
    \caption{Performance of NEIRD on MNIST.}
    \label{fig:nature source}
\end{figure*}

1) For the joint Gaussian model, the observation $X$ follows a multivariate Gaussian distribution $\mathcal{N}(0, K_X)$, and the source $S$ satisfies $S = HX + W$, where $H$ is a deterministic matrix and $W$ is multivariate Gaussian $\mathcal{N}(0, K_W)$, independent of $X$. We use an example in \cite{9844779} where $K_X$, $H$ and $K_W$ are:
\[
\small
K_X = \begin{bmatrix}
11 & 0 & 0.5 \\
0 & 3 & -2 \\
0.5 & -2 & 2.35
\end{bmatrix}, 
H = \begin{bmatrix}
0.0701 & 0.305 & 0.457 \\
-0.0305 & -0.220 & 0.671
\end{bmatrix},
\]
\[
K_W = \begin{bmatrix}
0.701 & -0.305 \\
-0.305 & 0.220
\end{bmatrix}.
\]
The distortion measure is $d(s, y)=\|s-y  \|_2^2$, and the reduced distortion measure becomes
\begin{align}
    \Bar{d}(x,y)=\mathrm{tr}(K_W)+\|Hx-y \|^2.
\end{align}
The iRDF is given by a semi-definite programming \cite{9844779}:\footnote{In \cite{9844779} an additional distortion constraint on the observation itself is imposed and here we remove it so as to focus on the iRDF only.}
\begin{align} \label{t_g}
    R(D) &= \min_{\Delta \in \mathcal{S}_m} \frac{1}{2} \log \frac{\mathrm{det}(K_X)}{\mathrm{det}(\Delta)}  \\
    \text{s.t.} & \quad \mathbf{0} \prec \Delta \preceq K_X,\nonumber\\
    & \quad \mathrm{tr} (H\Delta H^T) \leq D-\mathrm{tr}(K_W). \nonumber
\end{align}
Figure \ref{fig:fig1} displays the iRDF obtained via the NEIRD algorithm and it can be seen that the curve almost coincides with the theoretical solution \eqref{t_g}.

In addition to this example, we also study a higher dimensional example, where $K_X = 2\mathbf{I}_{120}$, $K_W = \mathbf{I}_6$ and $H$ is a $6\times120$ random matrix whose elements are i.i.d. taking values $\left\{-1, 0, 1 \right\}$ with probability $0.05, 0.90, 0.05$, respectively, to model a sparse correspondence between the source and the observation. Figure \ref{fig:fig2} compares the performance of the NEIRD algorithm and the theoretical solution. Again, highly accurate estimation is achieved.

2) For the binary classification model, $S$ is a binary source drawn from $\{0,1\}$ with equal probability, and the observation $X$ is a conditional Gaussian random variable:
\begin{align*}
    X \sim \mathcal{N}(A, \sigma^2), \text{if $S$ = 0;} \quad \sim \mathcal{N}(-A, \sigma^2), \text{if $S$ = 1}.
\end{align*}
Under the Hamming distortion (i.e., $d(x,y)=1$ if $x \neq y$ and $0$ if $x=y$), the iRDF is given by \cite{liu2021rate}:
\begin{align}
    R(D)=1-\frac{1}{2}\int_{-\infty}^{\infty}[N^+(x)+N^-(x)]h_b(g(x)) \mathrm{d}x, \label{binary}
\end{align}
where $N^+(x)$ and $N^-(x)$ denote the probability density functions of $\mathcal{N}(A, \sigma^2)$ and $\mathcal{N}(-A, \sigma^2)$ respectively, and $h_b(\cdot)$ is the binary entropy function.

Figure \ref{fig:fig3} compares the performance of the NEIRD algorithm and the theoretical solution, and their consistency is clearly evident.

\subsection{MNIST Dataset}\label{SCM}
In this subsection, we apply the NEIRD algorithm to the MNIST, which is a single-channel 28×28 grayscale dataset. First, we consider logarithmic loss as the distortion measure:
\begin{align}
    d(s,y) = \log \frac{1}{P_{Y}(s)}.
\end{align}
The logarithmic loss criterion has been commonly used in information theory and machine learning \cite{courtade2013multiterminal,cesa2006prediction}, and is a type of ``soft" criterion. Under the logarithmic loss, the reduced distortion measure \eqref{eqn:d-bar} becomes $H(S|Y)$, and by making a substitution $\tau = H(S) - D$, \eqref{ird} becomes
\begin{align}
    R(\tau) &= \min_{\substack{P_{Y|X}: \\ I(S;Y) \geq \tau }} I(X;Y),
\end{align}
which is exactly the information bottleneck problem \cite{tishby2000information,goldfeld2020information}.

In this case, $g_{\hat{\theta}}$ consists of three convolution layers, and a generative adversarial network is used to build $T_\theta$. Figure \ref{fig:fig4} displays the iRDF curve obtained via the NEIRD algorithm. Meanwhile, we also plot a theoretical solution claimed in \cite{DBLP:conf/iclr/KolchinskyTK19}, which is obtained under the extra assumption that $P_{S|X}$ is deterministic for the MNIST dataset, as $R(D) = \max\{H(S) - D, 0\}$, a linear segment.\footnote{We observe that the points generated by the NEIRD algorithm are clustered around end-points 
$(H(S), 0)$ and $(0, H(S))$. This is due to the essentially constant-slope behavior of the iRDF for the MNIST, noting that the NEIRD algorithm computes the iRDF for each given slope $\lambda \geq 0$.} In contrast, we emphasize that for our proposed NEIRD algorithm, there is no need to impose any assumption on the joint probability distribution $P_{S, X}$.

We further consider Hamming and cross-entropy distortion measures for $d(s, y)$, with their corresponding iRDF curves displayed in Figures \ref{fig:fig5} and \ref{fig:fig6}, respectively. To provide a benchmark, note that according to the data processing inequality, iRDF has a simple lower bound as follows:
\begin{align}
    R(D) = \min_{\substack{P_{Y|X}: \\ \mathbb{E}[\Bar{d}(X,Y)] \leq D }} I(X;Y) \geq \min_{\substack{P_{Y|S}: \\ \mathbb{E}[{d}(S,Y)] \leq D }} I(S;Y),
    \label{eq:lb}
\end{align}
where the right side is the standard RDF if the source encoder directly compresses $S$. In the case of Hamming distortion, the standard RDF has an explicit closed-form expression (see, e.g., \cite{alajaji2018introduction}) as $H(S) - D\log{(|\mathcal{S}|-1)} - h_2(D)$, if $0 \leq D \leq \frac{|\mathcal{S}|-1 }{|\mathcal{S}|}$ and zero otherwise, which is also plotted in Figure \ref{fig:fig5}, and is seen to be very close to the curve obtained via the NEIRD algorithm. These results provide further evidence of the accuracy and effectiveness of our proposed approach, which is purely data-driven without requiring any statistical knowledge.
\section{Conclusion}
In this work, we propose a data-driven approach to estimating the iRDF, and develop a neural network based algorithm called NEIRD for implementation. By leveraging a functional equivalence between the reduced distortion measure and a quadratic loss minimization problem, the key challenge for data-driven estimation of iRDF is resolved. Furthermore, a mapping approach is utilized to solve the variational problem, integrating together the reduced distortion measure estimator and the reproduction probability distribution, as implemented via a dual-layer nested neural network in the NEIRD algorithm. Numerical experimental results suggest that the proposed approach is a promising solution for various applications including remote sensing and goal-oriented communication.

\bibliographystyle{IEEEtran.bst}
\bibliography{IEEEabrv,references}
\end{document}